\documentclass[peerreview,12pt,letterpaper]{IEEEtran}

\newtheorem{thm}{Theorem}
\newtheorem{lemma}{Lemma}
\newtheorem{proposition}{Proposition}
\newtheorem{corol}{Corollary}

\newtheorem{remark}{Remark}
\newtheorem{defn}{Definition}

\renewcommand{\baselinestretch}{1.6}

\usepackage{epsf}
\usepackage{amssymb}
\usepackage{graphicx}
\usepackage{amsmath}
\usepackage[english]{babel}
\usepackage{mathrsfs}
\usepackage{subfigure}

\begin{document}

\title{\LARGE
Stochastic Geometry based Medium Access Games in Mobile Ad hoc Networks
}
\author{Manjesh Kumar Hanawal, Eitan Altman and Francois Baccelli
\thanks{Part of this work is accepted for presentation at Infocom 2012.}
\thanks{ Manjesh Kumar Hanawal and Eitan Altman are with INRIA and are located at the University of Avignon. Francois Baccelli  is with INRIA (Unite de Recherche de Rocquencourt) and ENS (Department d'Informatique).}

}

\date{}
\maketitle
\begin{abstract}
This paper studies the performance of Mobile Ad hoc Networks (MANETs) when
the nodes, that form a Poisson point process, selfishly choose their Medium Access Probability (MAP). We consider goodput and delay as the performance metric that each node is interested in optimizing taking into account the transmission energy costs. We introduce a pricing scheme based on the transmission energy requirements and compute the symmetric Nash equilibria of the game in closed form. It is shown that by appropriately pricing the nodes, the selfish behavior of the nodes can be used to achieve the social optimum at equilibrium.  The Price of Anarchy is then analyzed for these games. For the game with delay based utility, we bound the price of anarchy and study the effect of the price factor. For the game with goodput based utility, it is shown that price of anarchy is infinite at the price factor that achieves the global optima.
\end{abstract}

\begin{keywords}
Game Theory; Mobile Ad hoc Networks (MANETs); Pricing; Medium Access Control; Stochastic Geometry; Replicator Dynamics
\end{keywords}

\section{INTRODUCTION}
In this paper we study competition for network resources at
the medium access control (MAC) layer. It is well known that
computing the Nash equilibria in games is in general a hard problem.
Indeed, this problem falls into a class of problems introduced by
Christos Papadimitriou in 1994, called PPAD (Polynomial Parity
Arguments on Directed graphs). In view of this complexity, it
becomes attractive to identify classes of games for which
one can compute the equilibria at a low complexity.
We thus study a MAC game under some statistical assumptions
on the location of the nodes, which are on one hand reasonable in
many real scenarios, and on the other hand, allow for tractable and
in several cases, even explicit expressions for the Nash equilibria.

We consider slotted time, and assume that the transmitters are synchronized.
The basic assumptions on our model are
\begin{itemize}
\item The location of the transmitters at each time slot
forms a homogeneous Poisson point process (P.p.p);
\item Medium access is controlled using Aloha;
\item Transmission success is based on signal to interference and noise ratio (SINR) being larger than some threshold;
\item We assume saturated sources, i.e., every mobile has always a packet to send.
\end{itemize}

The Poisson assumption means that
\begin{itemize}
\item The number of mobiles in disjoint sets
are independent;
\item The number of mobiles in any given set follows a Poisson distribution. This class of point processes maximizes entropy. It is often used for modeling the location of users in e.g. mobile ad hoc networks.
\end{itemize}

We first consider the problem in which
each node is a player: it chooses a MAP so as to maximize its own performance metric.
We consider utility functions that model the tradeoff between
quality of service indicators (such as goodput expected delay) and power consumption related to the transmission. Our Aloha assumption on the MAC protocol leads to explicit expressions for the performance metrics of interest, which allow us to derive several interesting properties of the network. We leave the study of other, less tractable, CSMA type of MAC protocols for future study.

Our goals are to obtain symmetric Nash equilibria (SNE), possibly mixed, and to study
their properties \footnote{We shall compute the SNE but do not address the question of how equilibrium is reached}. By restricting to symmetric games, and by restricting to choices of MAP that are not functions of the locations of players or interferers, we have to consider the utility evaluated only at points where the strategies of all other players are the same, say $p$. We search for a $p$ such that if all use it then a player who deviates unilaterally will not improve its utility. Thus utility of a player can be viewed  as the utility of a player in a much simpler two person games, and any equilibria in the two player game is a symmetric equilibria in the original game.

We are interested in particular in the case
where the power consumption
disutility represents a pricing decision of the network that wishes
to determine a pricing that will induce an efficient equilibrium
(in terms of the achieved goodput).
Alternatively, the pricing may be taken such as to
maximize the network's revenues.

Our main findings are:
\begin{enumerate}
\item Considering the goodput as the quality of service to be maximized, we observe the
tragedy of the commons \cite{TragedyOfCommons}: the utility at equilibrium
is zero for large values of the pricing parameter.
Thus the price of anarchy is infinite. We show however that
there exists a pricing parameter for which the goodput
at equilibrium equals the one obtained under global
cooperative throughput maximization. We show that when each node uses the replicator dynamics to update its
MAP, the system converges to the SNE.
\item  Considering the delay as the quality of service to be minimized,
we observe that the price of anarchy is bounded for any price value.
Here too, there exists a pricing that results in  global optimum at equilibrium. We show that the SNE is not unique. The range of price parameters
for which two SNE exist is characterized.
\end{enumerate}

As the price of anarchy is infinite for certain price factors, it may seem therefore that high prices have a negative effect on the network performance. However if we consider instead the network performance or the monetary profit of the operator, we discover that pricing can induce an equilibrium for which these measures coincide with the global optimal values.

The Poisson assumption on the location of nodes allows us
to obtain utilities in a surprisingly simple explicit form,
which in turn allows us to obtain much insight on the
property of the equilibria and on the role of the pricing.
Such Poisson assumptions are often justified. For instance, in \cite{AndrewBacelliGantiAllerton2008_NewTractableApproach}, in the context of cellular networks, the authors obtain explicit expressions for the coverage and throughput with the Poisson assumption on base station locations. Their numerical investigation show that these expressions closely capture the real behavior of the network.

\noindent
\textbf{Related Work:}
There are several papers that model the nodes in the Aloha system as selfish users. In \cite{MacKenzieWickerInfocom2003_ALOHAGame} a game theoretic model is proposed to analyze the performance of the Aloha network. It is assumed that when a collision  occurs, some of the transmissions can be successful. The authors in \cite{HazerStephenIEEETON_AnalysisOfNEInRAG} study the performance of slotted Aloha with $n$ nodes that are selfish. Any collision results in the loss of the packet and there is a cost associated with unsuccessful transmissions. They characterize the equilibrium point as a function of $n$ and study its limiting behavior.
In \cite{AltmanElazouziTaniaWiOpt2003_ALOHAGameWithInformation}
slotted Aloha is considered with $n$ selfish users and the distributed choice of the retransmission probabilities is analyzed. It is shown numerically, that by adding a retransmission cost, the throughput at equilibrium equals the optimal team throughput.
Pricing is used in the context of the power control games in \cite{SaraMandaGoodIEEEJSAC_PricingPowerControl} to improve the performance at equilibrium.

All these paper do not take the geometric aspects of the node location into account , which is important in the wireless context. We consider in this paper the geometric aspects of the MANET and analyze the network performance at equilibrium. Also, unlike in all other papers, where simultaneous transmissions lead to transmission failure, we assume that the probability of successful transmission depends on SINR at the receiver. For games that take geometric aspects into consideration see \cite{AltmanKmarSinghIEEEInfocom2009_SpatialSINRGames}.

The paper is organized as follow. In Section \ref{sec:ModelAndSetup} the Poisson bipolar MANET model is set up and performance metrics of interest are discussed.  Section \ref{sec:TheTeamCase} considers the team case in which all nodes use the same MAP. Section \ref{sec:NonCoperativeCase} considers a scenario in which all nodes are selfish. A medium access game is defined among the nodes. Section \ref{sec:UtilityWithGoodput} studies this game with a utility function that is based on the goodput and a pricing based on the transmission energy costs. Section \ref{sec:UtilityWithDelay} studies the medium access game with another utility that is based on potential delay as performance metric and transmission costs. Section \ref{sec:PriceofAnarchy} analyzes the price of anarchy for these games. We end with some remarks and future work in Section \ref{sec:Conclusion}. All the proofs are given in Appendices.
\vspace*{-.1in}
\section{MODEL AND SETUP}
\vspace*{-.1in}
\label{sec:ModelAndSetup}
Consider the simplified mobile ad hoc network (MANET) model called the Poisson bipolar model proposed in \cite{BBM06}. Assume that each node follows the slotted version of the Aloha medium access control (MAC) protocol. Each dipole of the MANET consists of a transmitter and an associated receiver. Let the sequence $\{X_i,y_i\}_{i\geq 1}$ denote the location of the transmitters and receivers, where $y_i$ is the location of receiver associated with the transmitter at $X_i$. We assume that the transmitters $\Phi=\{X_i\}_{i\geq 1}$ are scattered in the Euclidian plane according to an homogeneous P.p.p of intensity $\lambda$. In this paper we consider a scenario in which each receiver is at a fixed distance $r>0$ from its  transmitter, i.e., $|X_i-y_i|=r$ for all $i$ \footnote{Our analysis continues to hold when the distance between transmitters and their receivers are i.i.d. }. Consider a snapshot of the location of the transmitters. Let the sequence $n=0,1,2,\cdots$ denote the common sequence of time slots with respect to which all nodes are synchronized. We associate with each node a multi dimensional mark that carries information about the MAC status and the fading condition at each time slot. We follow the notation of \cite{SGandWNVolII}[Chap. 17]. Let the sequence  $M_i=\{e_i(n),F_i(n)\}_{n\geq 0}$ denote the marks associated with node $i$, where
\begin{itemize}
    \item $e_i=\{e_i(n)\}_{n \geq 0}$ denotes the sequence of MAC status of node $i$. $e_i(n)$ is an indicator function that takes value $1$ if node $i$ decides to transmit in time slot $n$; otherwise it takes value zero. The random variable $e_i(n)$ are assumed to be i.i.d in $i$ and $n$, and independent of everything else.
     \item $F_i=\{F_i^j(n): j \geq 1\}_{n \geq 0}$ denotes the sequence of channel conditions between  the transmitter of node $i$ and all receivers (including its own receiver). It is assumed that channel conditions are i.i.d across the nodes and time slots, with a generic distribution denoted by $F$ with mean $1/\mu$.
    \item The marks are assumed to be independent in space and time.
\end{itemize}

The probability that the $i$th node transmits in time slot $n$ is $p:=\Pr\{e_i(n)=1\}=\mathbb{E}[e_i(n)]$ (Medium Access Probability (MAP)). This defines a pair of independent Poisson processes at each time slot $n$ , one representing transmitters $\Phi^1(n)=\{X_i,e_i(n)=1\}$ and the other non-transmitters $\Phi^0(n)=\{X_i,e_i(n)=0\}$ with intensities $p\lambda$ and $(1-p)\lambda$ respectively. All the nodes transmit at a fixed power denoted by $P$.

Let $l(x,y)$ denote the attenuation function between any two given points $x,y \in \mathbb{R}^2$. We assume that this function just depends on the distance between the points, i.e., $|x-y|$. With a slight abuse of notation we denote this function as $l(x,y)=l(|x-y|).$ We assume the following form for this attenuation function
\begin{equation}
\label{eqn:OPL3model}
l(x,y)=(A|x-y|)^{-\beta} \;\;\text{for} \;\;A>0 \;\; \text{and} \;\; \beta >2.
\end{equation}

A signal transmitted by a transmitter located at ${X_i}$ is successfully received in time slot $n$ if the SINR, at the receiver, is larger than some threshold $T$, i.e.,
\begin{equation}
\label{eqn:SINRThreshold}
SINR_i(n):=\frac{PF_i^i l(r)}{I_{\tilde\Phi^1(n)}(y_i)+W(n)} >T, \;\;
\end{equation}
where $W(n)$ denotes the thermal noise power at the receiver and $I_{\tilde\Phi^1(n)}$  denotes the shot noise of the P.p.p $\Phi^1(n)$ in time slot $n$, namely, $I_{\tilde\Phi^1(n)}(y_i)=\sum_{X_j \in \Phi^1(n)}PFl(|X_j-y_i|)$. We assume that the noise is an i.i.d process.

Consider a typical node at the origin, $X_0=0$ with mark $M_0(0)=(e_0(0), F_0(0))$ at $n=0$. The typical node is said to be covered in slot $n=0$ if (\ref{eqn:SINRThreshold}) holds given that it is a transmitter. Then the coverage probability of the typical node is
\begin{equation}
\label{TypicalNodeCoverageProbability}
\mathbf{P}^0\left \{\frac{PF l(r)}{I_{\tilde\Phi^1(0)}(y_0)+W(0)} >T \;\bigg  | \; e_0(0)=1\right \},
\end{equation}
where $\mathbf{P}^0$ denotes the Palm distribution \cite{SGandWNVolI}[Chap. I] of the stationary marked P.p.p $\tilde\Phi$. Note that due to time-homogeneity this conditional probability does not depend on $n$. By  using Slivnyak's theorem \cite{SGandWNVolI}, the coverage probability of a typical nodes when all other nodes use the same MAP is evaluated in \cite{Infocomm2010_BartekMuhlet_NonSlottedAloha}. Continuing the notation used in \cite{JSACSep2009_BaccelliBartekPaul} we denote this coverage probability (non-outage probability) as $p_c:=p_c(r, p\lambda , T)$. Consider a tagged node that uses MAP $p^\prime$. Then the tagged node is a transmitter with probability $p^\prime$ and a non transmitter with probability $(1-p^\prime)$. We refer to the product of the MAP and the coverage probability of the tagged node as \emph{goodput} and denote it as $
g(p^\prime, p):=p^\prime p_c(r, p\lambda , T).
$
We shall be interested in performance metrics of the form
\[\frac{g(p^\prime,p)^{1-\alpha}}{1-\alpha},\]
where $\alpha \geq 0$ denotes the fairness parameter. In this paper we consider the performance metric corresponding to $\alpha=0$ and $\alpha=2$. For $\alpha=0$, the performance metric corresponds to the goodput, and for $\alpha=2$, it corresponds to the negative of \emph{potential delay} introduced in \cite{BWSharingObjAlgo_TIT2002_MassRobe}, which is defined as the reciprocal of the rate. Potential delay may be interpreted as the delay incurred by a node in successfully delivering its packets at the receiver. Indeed, it has been shown in \cite{SGandWNVolII}[Sec. 17] that the \emph{local delay} is given as reciprocal of goodput when the node locations form an i.i.d process across time slots. In all other cases it is a lower bound. Local delay is defined as the number of time slots needed by a node to successfully transmit a packet (with retransmissions). We denote the potential delay as $t(p^\prime,p)=1/g(p^\prime,p)$.

In the following sections we consider two scenarios. First we assume that the nodes of the MANET cooperate, i.e., use the same MAP that is assigned to them in each time slot. The value of a MAP that optimizes the spatial network performance is evaluated. We then consider a game  scenario in which each  node is selfish and chooses a MAP that optimizes its own performance taking into account the transmission costs. We study the effect of the transmission costs on the network performance at equilibrium and look for a cost factor that results in improved spatial network performance at equilibrium.

\vspace*{-.1in}
\section{RATE CONTROL: THE TEAM CASE}
\label{sec:TheTeamCase}
\vspace*{-.1in}
In this section we assume that all the nodes belong to a single operator, and transmit at the MAP set by the operator. The following proposition immediately follows from \cite{BBM06}[Lemma 3.2]:
\begin{proposition}
Let each node in the the Poisson bipolar transmit with MAP $p$ and $F$ be Rayleigh distributed with mean $1/\mu$; then the goodput is
\begin{eqnarray}
\nonumber
g(p) &=&  p \exp \left\{
- 2 \pi \lambda p\int_0^{\infty}
\frac{u}{1 + l(r)/(T l(u))}
{\rm d} u \right\} \times \psi_{_W} (\mu T/Pl(r) ),
\end{eqnarray}
where $\psi_{_W}(\cdot)$ denotes the Laplace transform of the noise power $W$.
\end{proposition}
\begin{corol}
\label{corol:ZeroNoiseSuccessRate}
For F as above, zero noise power $W\equiv 0$, and the path loss model in (\ref{eqn:OPL3model}), the goodput of a typical node is
\begin{equation}
\label{eqn:TeamSuccessRate}
g(p)=p\exp\{-2\pi\lambda p r^2 T^{2/\beta}K (\beta)\},
\end{equation}
where $K(\beta)=\frac{\Gamma (2/\beta)\Gamma (1-2/\beta)}{\beta}$ and $\Gamma (x)=\int_{0}^\infty z^{x-1}e^{-z}.$
\end{corol}
Hence forth we adopt the assumptions of Corollary  \ref{corol:ZeroNoiseSuccessRate} in all the subsequent calculations. However, our results hold when $W$ has any distribution with finite mean, as it appears as a constant multiplicative factor in (\ref{eqn:TeamSuccessRate}). For notational convenience we write $C:=C(\beta, T, r)= 2\pi r^2 T^{2/\beta}K (\beta)$.

The operator is interested in optimizing the social performance of the network. In particular, we assume that the operator aims at maximizing the mean goodput per unit area or minimizing the mean delay per unit area. The performance seen by a typical node can be used to derive the spatial performance of the Poisson bipolar MANET. Campbell's formula \cite{SGandWNVolI}[Sec. 2.1.2] for stationary Poisson point processes ensures that the performance experienced by a typical node is also that of the spatial average performance of the Poisson MANET. The mean goodput per unit area is then the product of the intensity of the P.p.p and the goodput, i.e., $\lambda  g(p)$. This quantity is referred to as the density of success and denoted by $d_{suc}(r, p\lambda, T)$ in \cite{SGandWNVolII}[Chap. 16]. We denote this term  simply as $d_s(p)$. Similarly, the mean delay per unit surface area is given by $\lambda t(p)$. We denote this spatial performance metric as $d_t(p)$ and refer to it as spatial delay density. Note that the spatial delay density is the  reciprocal of the density of success multiplied by a factor $\lambda^2$. Hence the MAP that maximizes the $d_s(p)$ also minimizes $d_t(p)$. The MAP that optimizes the density of success is given in \cite{SGandWNVolII}[Prop 16.8] and \cite{SGandWNVolII}[Corol. 16.9]
: \begin{proposition}
\label{prop:GlobalMaximaDensityofSuccess}
Under the assumption of Corollary \ref{corol:ZeroNoiseSuccessRate}, the MAP that maximizes the density of success and minimizes the density of delay is given by
\begin{equation}
\label{eqn:OptimalMAPDensityofSucess}
p_m=\min\{1,1/\lambda C\},
\end{equation}
and the corresponding optimal density of success is given by
\begin{equation}
\label{eqn:GlobalMaximumDensitySuccess}
d_s(p_m)=\left\{
  \begin{array}{ll}
    1/(e\lambda C), & \hbox{if\;\;} \lambda C >1 \\
    \lambda\exp\{-\lambda C\}, & \hbox{if\;\;} \lambda C \leq 1,
  \end{array}
\right.
\end{equation}
and the corresponding optimal delay density is given by
\begin{equation}
\label{eqn:GlobalOptimalDensityOfDelay}
d_t(p_m)=\left\{
  \begin{array}{ll}
    \lambda^2 e C, & \hbox{\;\; if \;\; } \lambda C >1\\
    \lambda \exp\{\lambda C\}, & \hbox{\;\; if \;\; } \lambda C \leq 1.
  \end{array}
\right.
\end{equation}
\end{proposition}
\vspace*{-.1in}
\section{RATE CONTROL: THE NON-COOPERATIVE CASE}
\vspace*{-.1in}
\label{sec:NonCoperativeCase}
In this section we assume that each node of the Poisson bipolar MANET is a selfish player. We consider a non-cooperative case, and model it as a game with an infinite number of players as follows. We use node and player interchangeably.

Each node can take two actions: transmit (T) or no-transmit (NT). We assume that when taking the decision whether to transmit or not, a player does not
know the positions of other mobiles nor the level of the SINR. A mixed strategy chosen by a node corresponds to its MAP. A player chooses its MAP once for all and uses always the same MAP. The choice is done without knowledge of the realization of the interference.

We shall restrict to symmetric Nash equilibria (SNE) , in which all nodes use the same MAP $p$ at equilibrium. The utility of a {\it tagged} player can then be written as a function of his strategy $p^\prime$ and of the strategy $p$ of all the others. We denote the utility of the tagged player as $U(p^\prime,p)$\footnote{The game has infinitely many players, hence the utility should be defined on a infinite product strategy space. However, since we are restricting to the case where all players other than tagged node use the same strategy $p$, we write utility as a function of two arguments}. The objective of each player is to maximize its utility. Let  $U(p^*,p^*)$ denote the value of the utility function at SNE $p^*$.
\begin{defn}
\label{defn:SymmetricNashEquilibrium}
$p^* \in [0\;1]$ is said to be a
symmetric Nash equilibrium if for any node the following holds
\[U(p^* , p^*) = \max_{p\in [0\; 1]} U( p,p^*).\]
\end{defn}

In the rest of this section we look for the appropriate utility functions that characterize the performance of an individual player. We begin with computing the goodput of a tagged node. The next lemma follows from Lemma 3.1 and Corollary 3.2 in \cite{BBM06}.

\begin{lemma}
Consider a Poisson MANET with the assumptions in Corollary \ref{corol:ZeroNoiseSuccessRate}. Let a tagged node transmit with MAP $p^\prime$, while all other players transmit with MAP $p$. Then the goodput of the tagged node is
\begin{equation}
g(p^\prime,p)=p^\prime \exp\{-p\lambda C\}.
\end{equation}
 where
$C =2\pi r^2 T^{2/\beta} K(\beta)$ is defined earlier.
\end{lemma}

\noindent
We note the following:
\begin{enumerate}
\item
The goodput is monotone increasing
in $p'$. Hence, if the objective of each mobile is
to maximize its goodput, then the only equilibria
is $p=1$ for all nodes.
\item
Under the conditions of the above lemma,
the expected transmission energy consumption of the tagged node per time slot is proportional to $p^\prime$ and does not depend on the MAP $p$ of the other nodes. Thus the ratio between the goodput and expected energy consumption does not depend on $p^\prime$.
We conclude that any $p$ is an  equilibria
when the criterion of each node is that of
minimizing the ratio between goodput and energy.
\end{enumerate}
\noindent
Keeping the above remarks in view and noting that
the utility related to goodput is not a linear function of the MAP, see
for example \cite{emodel2}, we shall be interested in utilities\footnote{With some abuse of notation we use the notation $U(p)$ to denote the expected system utility when the same $p$ is used by all players. We consider utilities which attain optimal values for some $p$ and denote the latter, again with some abuse of notation, as $p_m$. It will be clear from the context if $p_m$ is an optimizer of the spatial performance metric or the team utility.} of the form
\begin{equation}
\label{eqn:GeneralUtility}
U( p^\prime,p)
=f( p^\prime,p) - w(p')
\end{equation}
for each player, where $f$ gives the performance measure of interest and $w$ is related to dis-utility.
\noindent
We consider a dis-utility based on the expected transmission energy costs incurred by each node. Let $\rho$ denote the price per unit transmission energy for each node. Then the expected transmission cost for the tagged node that uses MAP $p^\prime$ is $\rho p^\prime P$. We define $w(p^\prime)=\rho p^\prime \footnote{without loss of generality we assume $P=1$}$ for a node that uses MAP $p^\prime$.

Assume that the function $f$ is a concave function in $p^\prime$ and continuous in $p$; then the arguments in \cite{rosen}[Thm. 1] can be used to show the existence of SNE. We state this result in the following lemma. The proof is given in Appendix \ref{app:ProofNEExistance}.
\begin{lemma}
\label{lma:NEExistance}
Assume that the utility function $U(p^\prime,p)$ is concave in $p^\prime$ and continuous in $p$.  Then SNE exist.
\end{lemma}

In this paper we assume that $\rho$ is set by a central agent who can regulate its value. The aim of the central agent is to optimize the performance of the network at equilibrium. We can also assume that the aim of the central agent is to maximize the network revenues at equilibrium. We refer to $\rho$
as the price factor.

In the next two sections we consider two utility functions defined in term of goodput and potential delay as the performance measure. We evaluate the MAP at equilibrium and the corresponding system performance. This system performance is then optimized by searching for the best price factor. The best achievable system performance, at equilibrium, is then compared with that evaluated when nodes act as a team.
\vspace*{-.1in}
\section{GOODPUT BASED UTILITY}
\vspace*{-.1in}
\label{sec:UtilityWithGoodput}
Assume that each node of the MANET is interested in maximizing its goodput taking into account the energy cost it incurs. Then by taking $f(p^\prime,p)=g(p^\prime,p)$ in (\ref{eqn:GeneralUtility}), we define the utility as
\begin{equation}
\label{eqn:UtilitySuccessRate}
U(p^\prime,p)=g(p^\prime,p)-\rho p^\prime=p^\prime \left \{\exp\{-p\lambda C\}-\rho \right \}.
\end{equation}
The objective of a each node is to choose a MAP that maximizes its utility, i.e.,
\[p^\prime \in \displaystyle \text{argmax}_{0\leq p^\prime \leq 1}\;\; U(p^\prime,p).\]
This utility function is a linear function in $p^\prime$ and convex in $p$. Then from Lemma \ref{lma:NEExistance}, a SNE exists. We proceed to calculate the SNE of this game.

When $\rho \geq 1$ the slope of the utility of the tagged node is non positive irrespective of the MAP of the other nodes. Then the  optimal strategy for each node is to choose $p=0$, which is a dominant strategy and hence the SNE. When $\rho < 1$ consider the following two cases.\\
Assume $\rho \geq \exp\{-\lambda C\} $: In this case the slope of the tagged node is always positive. Then the optimal strategy for the tagged node is to choose $p=1$ irrespective of the MAP chosen by the other nodes. Thus $p=1$ is a dominant strategy and hence also is the SNE.\\
Assume $\rho < \exp\{-\lambda C\} $: If each node other than the tagged node chooses a MAP such that
\begin{equation}
\label{eqn:EquilibriumMAPSuccessRate}
\exp\{- p \lambda C\}=\rho,
\end{equation}
then the utility of the tagged node in (\ref{eqn:UtilitySuccessRate}) is set to zero and is not affected by its strategy, i.e., the taged node becomes indifferent to its own strategy. Further, if any of the nodes deviates from the MAP that satisfies (\ref{eqn:EquilibriumMAPSuccessRate}), it will not gain anything given that all other nodes set their MAP value as in  (\ref{eqn:EquilibriumMAPSuccessRate}). Hence the MAP satisfying (\ref{eqn:EquilibriumMAPSuccessRate}) constitutes the SNE. We summarize the above observations in the following proposition.

\begin{proposition}
\label{prop:GoodputequilibriumMAP}
For any given $C, \lambda$, and $\rho >0$
\begin{itemize}
\item if $\rho \geq 1$ then $p^*=0$ is the SNE;
\item if $\exp\{-\lambda C \} \geq \rho$ then $p^*=1$ is the SNE;
\item if $\exp\{-\lambda C \} < \rho < 1$   then
$p^*=\frac{-\log \rho}{\lambda C}$ is the SNE.
\end{itemize}
\end{proposition}
\noindent
The goodput of each node at equilibrium is given by
\begin{equation}
\label{eqn:EquilibriumSuccessRate}
g(p^*,p^*)=\left\{
  \begin{array}{ll}
    0, & \hbox{\;if\;} \rho \geq 1\\
    \exp\{-\lambda C\}, & \hbox{\;if\;} \rho \leq \exp\{-\lambda C\}\\
    \frac{-\rho\log \rho}{\lambda C}, & \hbox{\;if \;} \exp\{-\lambda C \} < \rho < 1.
  \end{array}
\right.
\end{equation}
With the expression for goodput at equilibrium, we can look for the value of $\rho$ that maximizes it.
We assume that the objective of the central agent is to maximize the density of success at equilibrium. Then the optimization problem of the central agent is given by
\begin{equation}
\label{eqn:SuccessRateOptimization}
\begin{aligned}
& \underset{\rho}{\text{maximize}}
& & \frac{-\lambda \rho\log \rho}{\lambda C} &&& \text{subject to}
&&&& \exp\{-\lambda C\}<\rho < 1.
\end{aligned}
\end{equation}
The objective function in (\ref{eqn:SuccessRateOptimization}) is a concave function of $\rho>0$ attaining its maximum at
$\rho=1/e.$ If $\lambda C>1$ the optimal price factor lies within the constraint set and the operator can set $\rho^*=1/e$. Suppose $\lambda C \leq 1$; then the objective function is decreasing in the interval $\exp\{-\lambda C\}<\rho <1$ and the maximum is attained at $\rho^*=\exp\{-\lambda C\}.$  The maximum density of success at equilibrium with the optimal price factor is
\begin{equation}
\label{eqn:EquilibriumMaximumSuccessRate}
d_{s}(p^*,p^*)=\left\{
  \begin{array}{ll}
    1/(e\lambda C), & \hbox{if \;\;} \lambda C>1 \\
    \lambda\exp\{-\lambda C\}, & \hbox{if\;\;} \lambda C\leq 1.
  \end{array}
\right.
\end{equation}
Comparing (\ref{eqn:EquilibriumMaximumSuccessRate}) with the global optimal density of successful transmissions given in (\ref{eqn:GlobalMaximumDensitySuccess}), we have the following result.

\begin{thm}
The global optimal density of successful transmissions can be attained at equilibrium by setting the pricing factor $\rho$ as follows:
\begin{equation}
\label{eqn:OptimalRhoSuccessRate}
\rho^*=\left\{
  \begin{array}{ll}
    1/e, & \hbox{if \;\;}  \lambda C >1\\
    \exp\{-\lambda C\}, & \hbox{if \;\;} \lambda C \leq 1.
  \end{array}
\right.
\end{equation}
\end{thm}
Also, note that with the pricing factor $\rho^*$ in (\ref{eqn:OptimalRhoSuccessRate}), the MAP of each node at equilibrium is the same as that achieving the global optimum given in (\ref{eqn:OptimalMAPDensityofSucess}).
Thus by pricing appropriately, one can use the selfish behavior of the players to reach an equilibrium at which the global optimal performance is attained.

\noindent
\textbf{Replicator Dynamics:} In this subsection we briefly discuss how we can use tools from the population games and evolutionary dynamics to study aggregate behavior in the network. We can think the set of nodes in the Poisson bipolar MANET as a single population, where each node can take two actions. Let each node decide to transmit with probability $p$, then each node will be in state of transmission with probability $p$ and in the state of non-transmission with probability probability 1-p. Then, by our assumption that nodes decisions are i.i.d in time and space, a fraction $p$ of the population will be in the state of transmission and the other (1-p) fraction in the non-transmission state. In short we denote the state of the network as $p$ when every node transmits with probability $p$.

One of the most frequently used dynamics to describe the evolution of behavior in population games is the replicator dynamics \cite{WHSandholm}. It describes the evolution of the fraction of the population in each state. In terms of population games, we can interpret the probability of successful transmission with the mean transmission energy costs as the fitness function for each node. Note that in our medium access game the fitness of a node that is not transmitting is zero. Let us denote the utility/fitness of a node that chooses action T while the rest of the population is in state $p^\prime$ as $U(T,p^\prime)$. The fraction of nodes that uses action $T$ evolve according to the replicator dynamics as follows \cite{WHSandholm}[Ch. 5,6]:
\[{\rm d}p(t)/{\rm d}t=p(t)[U(T,p(t))-\overline{U}(t)]\]
where $\overline{U}(t)$ denotes the mean fitness of the population. Taking (\ref{eqn:UtilitySuccessRate}) as the fitness function, the replicator dynamics is given by
\begin{eqnarray*}
\label{eqn:DynamicsReplicatorProbSucc}
\frac{{\rm d} p(t) }{{\rm d} t}&=&p(t)\left \{{e^{-\lambda C p(t)}-\rho}-p(t){e^{-\lambda C p(t)}-p(t)\rho}\right \}
=p(t)(1-p(t))\left \{e^{-\lambda C p(t)}-\rho \right \}.
\end{eqnarray*}
From the above equation the stationary point is clearly 0 or 1 or $\frac{\log \rho}{\lambda C}$ depending on the value of $\rho$, which is in agreement with Proposition \ref{prop:GoodputequilibriumMAP}. Figure \ref{fig:Dynamics} shows the convergence of  replicator dynamics to the SNE starting from any interior point.

\section{DELAY BASED UTILITY}
\label{sec:UtilityWithDelay}
In this section we define the utility of each player in terms of the delay associated in delivering the packet successfully at its receiver, and the transmission costs. As earlier, let a tagged node incur a cost of $\rho$ units per unit energy dissipated. By taking $f(p^\prime,p)=-t(p^\prime,p)$ in (\ref{eqn:GeneralUtility}), we get the following utility for the tagged node:
\begin{equation}
\label{eqn:DelayUtility}
U(p^\prime,p)=-\frac{1}{p^\prime \exp\{-p\lambda C\}} - p^\prime \rho.
\end{equation}
where $p^\prime$ is the MAP of the tagged node, and $p$ is the MAP used by all other nodes.

The objective of each node is to choose a MAP that maximizes its utility function (\ref{eqn:DelayUtility}). Or equivalently it can be given by the following minimization problem:
\begin{equation*}
\begin{aligned}
& \underset{p^\prime}{\text{minimize}}
& & \frac{\exp\{p\lambda C\}}{p^\prime } + p^\prime \rho &&& \text{subject to} \; 0\leq p^\prime \leq 1.
\end{aligned}
\end{equation*}
For a given value of $p$, the utility function in (\ref{eqn:DelayUtility}) is a concave function in $p^\prime $ and continuous in $p$. Then by Lemma \ref{lma:NEExistance} SNE exist. We proceed to calculate the SNE by a direct computation.

\noindent
Differentiating the utility function with respect to $p^\prime$, equating to zero, and simplifying, we get
$p^\prime={\exp\{p\lambda \overline C\}}/\sqrt\rho$,
where $\overline{C}=C/2.$ This equation gives the best response of the tagged node when all other nodes use MAP $p$. If there exists a $p \in [0\;1]$ such that
$p^\prime=p$, then $p$ is the SNE of the game. Hence we look for the conditions when the following fixed point relation holds

\begin{equation}
\label{eqn:DelayUtilityFixedPoint}
p=\frac{\exp\{p\lambda \overline{C}\}}{\sqrt\rho}.
\end{equation}

\begin{lemma}
\label{lma:LambertCondition}
If  $\lambda \overline{C}e > \sqrt \rho$, $p^*=1$ is the unique SNE. If $\lambda \overline{C}e\leq \sqrt \rho$, the SNE is
\begin{equation}
p^*=\min\left \{\frac{-1}{\lambda \overline{C}}\mathbb{W}\left(-\frac{\lambda \overline{C}}{\sqrt \rho}\right),1\right \},
\end{equation}
where $e$ denotes the base of the natural logarithm and  $\mathbb{W}(\cdot)$ denotes the Lambert function \cite{LambertFunction}.
\end{lemma}
\begin{proof}
By using the relation $e^x \geq xe$ for all $x \geq 0$, it is easy to see that $p^*=1$ is the SNE when $\lambda \overline{C}e > \sqrt \rho$. Assume $\lambda \overline{C}e \geq \sqrt \rho$; rearranging the terms in (\ref{eqn:DelayUtilityFixedPoint}), we can write
$-p\lambda \overline{C}\exp\{-p\lambda \overline{C}\}=\frac{-\lambda \overline{C}}{\sqrt \rho}$.
Using the definition of the Lambert function \cite{LambertFunction}, we obtain
\begin{equation}
\label{eqn:LambertSolution}
p=\frac{-1}{\lambda \overline{C}}\mathbb{W}\left(-\frac{\lambda \overline{C}}{\sqrt \rho}\right).
\end{equation}
\end{proof}
\noindent
Let us briefly recall the properties of the Lambert function in the footnote below \footnote
{\begin{itemize}
\item The Lambert function is two-valued in the interval $[-1/e,\;0]$. The two branches of the Lambert function denoted as $\mathbb{W}_0(\cdot)$ and $\mathbb{W}_{-1}(\cdot)$ meet at $-1/e$ and the take value $-1$ at this point.
\item In the interval $[-1/e, \; 0]\;$  $\mathbb{W}_0(\cdot)$ is a continuous and increasing function taking value in $[-1, \; -\infty]$.
\item In the interval $[-1/e, \; 0]\;$  $\mathbb{W}_{-1}(\cdot)$ is a continuous and  decreasing function taking value in $[-1 ,\; \infty]$.
\end{itemize}
}.
With this explicit expression for the SNE we can characterize the effect of the price factor on the SNE. For some values of $\rho$, the resulting SNE is not unique as shown in the following lemma. For notational simplicity we write $\rho_t=(e\lambda \overline{C})^2$.

\begin{lemma}
\label{lma:DealyTwoEquilibriumCondn}
Assume $\lambda \overline{C} \geq 1$; 
then there exists a $\rho_{-1}\geq \rho_t $ such that for all $\rho$ satisfying $\rho_{-1}\geq \rho \geq \rho_t$ two SNE are possible. One resulting from the $\mathbb{W}_0$ and the other from the $\mathbb{W}_{-1}$ branch of the Lambert function.
\end{lemma}

\noindent

Note that for for values of  $\rho \in [\rho_t, \; \rho_{-1}] $, the equilibrium MAP computed on the $\mathbb{W}_{-1}$ is larger than that computed on the $\mathbb{W}_0$ branch. If nodes reach the equilibrium that is resulting from the $\mathbb{W}_{-1}$ branch, then they will be transmitting aggressively at equilibrium. Later we will see that this leads to inefficiency (\ref{rmk:BadEquilibrium}).

\begin{lemma}
\label{lma:DelayUniqueNE}
Assume $\lambda \overline{C} < 1$, there exists a $\rho_0 \geq \rho_t$ such that for all $\rho$  satisfying  $\rho_0 \geq \rho \geq \rho_t$, the SNE is $p^*=1$. For all $\rho$ satisfying $\rho \geq \rho_0$ the SNE is unique and lies on the $\mathbb{W}_0$ branch  of the Lambert function.
\end{lemma}
Lemma \ref{lma:LambertCondition} together with \ref{lma:DealyTwoEquilibriumCondn} and \ref{lma:DelayUniqueNE} completely characterize all possible equilibria. We summarize all the SNE resutls in Table \ref{tab:DelaySNE}. Figures \ref{fig:DelayLambert1} and \ref{fig:DelayLambert2} illustrate the Lambert function and its properties used in the proof of Lemma \ref{lma:DealyTwoEquilibriumCondn} and \ref{lma:DelayUniqueNE}. In Figure \ref{fig:DelayLambert1} the $y$ axis marked with a double arrow denotes the region in which two equilibrium points occur. \\
\noindent
\textbf{Stability of SNE:}
In Figure \ref{fig:Best Response} we plot the best response of the tagged node and that of all nodes against the tagged node. This example demonstrates the existence of two SNE. The smaller SNE among the two results from the principal branch of the Lambert function in (\ref{eqn:LambertSolution}), and the larger one from the $\mathbb{W}_{-1}$ branch. When two equilibrium points exist, we denote the SNE resulting from the $W_0$ branch of the Lambert function as $p_0^*$ and that resulting from the $\mathbb{W}_{-1}$ branch as $p_{-1}^*$. When we do not need to make this distinction or it is unique we write the SNE as $p^*$. From Figure \ref{fig:Best Response} we see that, at the equilibrium point $p_0^*$, a slight increase in the MAP $p$ results in a decrease in $p^\prime$. This is a stabilizing behavior and we conclude that $p_0^*$ is stable. In contrast, at the equilibrium point $p_{-1}^*$ a slight increase in $p$ is seen to cause an increase in the value of $p^\prime$. Thus the second equilibrium point is unstable.\\
\noindent
\textbf{Optimal Pricing:}
Assume that $p^*$ is an equilibrium point for a given value of $\rho$ that satisfies the conditions in Lemma \ref{lma:DealyTwoEquilibriumCondn} or \ref{lma:DelayUniqueNE} for a given value of $\lambda$ and $\overline C$. The potential delay experienced by the packets of a tagged node at equilibrium can be computed as

\begin{eqnarray}
\nonumber
\lefteqn {t(p^*,p^*)=\frac{1}{p^*\exp\{-p^* \lambda C\}}} \\
\label{eqn:DelaySplitC}
&=& \frac{1}{p^*\exp\{-p^* \lambda \overline C \} \exp\{-p^* \lambda \overline C\}}\\
\label{eqn:DealyUSeFixedPointDefn}
&=& \frac{\sqrt \rho}{\exp\left\{\mathbb{W}\left(-\lambda\overline C / \sqrt \rho \right)\right\}}\\
\label{eqn:DelayBringStdLambertFromat}
&=& \frac{\sqrt \rho \mathbb{W}\left(-\lambda\overline C / \sqrt \rho \right)}{\mathbb{W}\left(-\lambda\overline C / \sqrt \rho \right)\exp\left\{\mathbb{W}\left(-\lambda\overline C / \sqrt \rho \right)\right\}}\\
\label{eqn:DealyAtEquilibrium}
&=& \frac{\sqrt \rho \mathbb{W}\left(-\lambda\overline C / \sqrt \rho \right)}{(-\lambda\overline C / \sqrt \rho)}=-\frac{ \rho \mathbb{W}\left(-\lambda\overline C / \sqrt \rho \right)}{\lambda\overline C}
\end{eqnarray}
In above chain of equalities the relation $C=2 \overline C$ is used in (\ref{eqn:DelaySplitC}). (\ref{eqn:DealyUSeFixedPointDefn}) follows from (\ref{eqn:DelayUtilityFixedPoint}). Equation (\ref{eqn:DealyAtEquilibrium}) follows by applying the definition of the Lambert function to the denominator term in (\ref{eqn:DelayBringStdLambertFromat}) and rearranging.
From \ref{eqn:DealyAtEquilibrium} we see that when the SNE is not unique when the Lambert function takes two values for some price factors. Delay experienced by the tagged node is larger when the SNE results from the $\mathbb{W}_{-1}$ branch of the Lambert function.

With the expression for the delay of a tagged node at equilibrium, one can look for the value of the price factor that minimizes the delay experienced by each node at equilibrium. We assume that the objective of the central agent
is to minimize the average delay per unit area, i.e., spatial delay density, at equilibrium.  The optimization problem of the central agent is:
\begin{equation}
\label{eqn:DelayOptimization}
\begin{aligned}
& \underset{\rho}{\text{minimize}}
& & -\frac{ \rho \mathbb{W}\left(-\lambda\overline C / \sqrt \rho \right)}{ \overline C}
& && \text{subject to}
&&&& \hspace{-.4cm} \sqrt\rho \geq \lambda \overline C e
&&&&&  \hspace{-.2cm}\text{and}
&&&&&& \hspace{-.4cm} -\mathbb{W}\left(-\lambda\overline C / \sqrt \rho \right) \leq \lambda \overline C .
\end{aligned}
\end{equation}
The first constraint in this optimization problem results from Lemma \ref{lma:LambertCondition} and the second constraint is required to ensure that the resulting value of $p$ in ($\ref{eqn:DelayUtilityFixedPoint}$) lies in the interval $[0\; 1].$

Let $h(\rho):=-\rho \mathbb{W}\left(-\lambda\overline C / \sqrt \rho \right)$ denote the objective function in the above optimization problem without the multiplicative factor $\overline C$. $h(\rho)$ is defined for $\sqrt \rho \geq e \lambda \overline C$. In the following lemma we state some of its properties.

\begin{lemma}
\label{lma:DelayObjectiveProperties}
On the principal branch $\mathbb{W}_0$, $h(\rho)$ is a quasi convex function in $\rho$ and the global minimum is attained at $\rho^*=4 e (\lambda \overline C)^2$. On the $\mathbb{W}_{-1}$ branch $h(\rho)$ is a monotonically increasing function taking value $(e \lambda \overline C)^2$ at $\rho=(e \lambda \overline C)^2$.
\end{lemma}
Consider the optimization problem in (\ref{eqn:DelayOptimization}) on the $\mathbb{W}_0$ branch of the Lambert function. The value of $\rho^*$ at which $h(\rho)$ achieves minima satisfies the first condition in (\ref{eqn:DelayOptimization}) as ${\rho^*}=4e(\lambda \overline C)^2 \geq e^2(\lambda \overline C)^2.$ The value of $p$ in equation (\ref{eqn:DelayUtilityFixedPoint}) at $\rho=\rho^*$ is
$1/(2\lambda \overline C)$. This implies that whenever $2\lambda \overline C= \lambda C >1$, the resulting value of $p$ lies inside $(0\;1),$ thus satisfying the second condition of the optimization problem. Hence under the assumption $\lambda C >1$, the global minimizer of $h(\rho)$ lies in the constraint set of $(\ref{eqn:DelayOptimization})$ and the value of the objective function at this point is
\begin{equation}
\label{eqn:DelayAtEquilibriumOptimal}
-\frac{\rho^* \mathbb{W}(-\lambda\overline C /\sqrt {\rho^*})}{\lambda\overline C}=-\frac{(4e(\lambda\overline C)^2)(-1/2)}{\lambda\overline C}=e\lambda C.
\end{equation}
When $\lambda C \leq 1$, any value of $\rho$ such that $\rho < \rho^*$ violates the second condition, as $-\mathbb{W}(-\lambda\overline C /\sqrt {\rho})>1/2 $. From Lemma \ref{lma:DelayObjectiveProperties}, the minimum value of $h(\rho)$ is achieved at a $\rho$ satisfying $-\mathbb{W}(-\lambda \overline C/\sqrt \rho)=\lambda \overline C$. This implies that under the condition $\lambda C \leq 1$ the value of the SNE is $p^*=1$ and the delay experienced by each node is given by $\exp\{-\lambda C\}$. We summarize these observations in the following result.

\begin{thm}
The value of the price factor $\rho^*$ that minimizes the delay at equilibrium is as follows:
\begin{equation}
\label{eqn:DelayOptimalRhoValues}
\rho^*=\left\{
  \begin{array}{ll}
    (2\lambda \overline C\sqrt e)^2  & \hbox{\;\;if \;\;} \lambda C >1\\
    -2\mathbb{W}(-\lambda \overline C/\sqrt {\rho^*})=\lambda C, & \hbox{\;\;if \;\;} \lambda C \leq 1
  \end{array}
\right.
\end{equation}
and the corresponding delay at equilibrium is
\begin{equation}
\label{eqn:DelayAtEquilibriumOptimalWithOptRho}
d_t(p^*,p^*)=\left\{
  \begin{array}{ll}
    \lambda^2 e C  & \hbox{\;\;if \;\;} \lambda C >1\\
    \lambda \exp\{\lambda C\}, & \hbox{\;\;if \;\;} \lambda C \leq 1.
  \end{array}
\right.
\end{equation}
\end{thm}

Comparing (\ref{eqn:DelayAtEquilibriumOptimalWithOptRho}) and (\ref{eqn:GlobalOptimalDensityOfDelay}) we conclude the following result:
\begin{proposition}
\label{prop:DelayGlobalAndEquilibriumOptimal}
The spatial delay density in the game problem at equilibrium equals the global optimal spatial delay density, i.e., $d_t(p_m)=d_t(p^*,p^*)$,  if the price factor is chosen as in (\ref{eqn:DelayOptimalRhoValues}).
\end{proposition}

Again, by appropriately pricing the nodes, the selfish behavior can be used to attain the global optimal performance in the game problem. If the SNE is not unique at the optimal price factor $\rho^*$, then one needs to ensure that the nodes reach the equilibrium that is computed on the principal branch of the Lambert function to realize the global optimal performance at equilibrium. Indeed, if $\lambda C >1$  and $-\mathbb{W}_{-1}(-1/2\sqrt e) \leq \lambda \overline C$  or equivalently $\rho_{-1} \geq 4e (\lambda \overline C )^2$, at the optimal price $\rho^*$, an equilibrium point exists on the $\mathbb{W}_{-1}$ branch of the Lambert function. It would be interesting to learn about the way to reach a favorable equilibrium at the optimal price factor. However we do not pursue this question in this paper.
\begin{remark}
\label{rmk:BadEquilibrium}
If $\lambda \overline C >1$ and optimization is restricted to the $\mathbb{W}_{-1}$ branch in (\ref{eqn:DelayOptimization}), then by Lemma \ref{lma:DelayObjectiveProperties}, the objective function is minimized by choosing a $\rho$ satisfying $\sqrt \rho = e\lambda \overline C$, which results in the equilibrium probability $1/(\lambda \overline C)$. The spatial density of delay at this equilibrium point is given by
\[-\frac{ \rho \mathbb{W}\left(-\lambda\overline C / \sqrt \rho \right)}{\lambda\overline C}=\frac{(e\lambda\overline C)^2}{\lambda \overline C}=(e/2)e\lambda C.\]
Comparing this value with (\ref{eqn:GlobalOptimalDensityOfDelay}), we see that spatial density of delay increased by a factor of $e/2$ by the selfish behavior of the nodes.
\end{remark}
\vspace*{-.1in}
\section{PRICE of ANARCHY}
\label{sec:PriceofAnarchy}
\vspace*{-.1in}
In this section we study the degradation in the network performance due to a selfish behavior of the nodes. The Price of Anarchy (PoA) compares the social utility at the worst equilibrium with the optimal social utility \cite{PriceOfAnarchy}. For our Poisson bipolar MANET with infinitely many players, we define the PoA as the ratio of the optimal spatial average performance that can  be achieved, to the spatial average performance at the worst SNE. Recall that we denoted the system utility by $U(p)$ when we considered the team problem, with each node using the same MAP $p$. For the game problem we denoted the utility of a tagged node by $U(p^\prime,p)$. In the game problem the spatial average performance at equilibrium  is evaluated by multiplying the utility of the tagged node and the intensity of the P.p.p . Then the PoA is given by

\begin{equation}
\label{eqn:POADefinition}
PoA=\frac{\max_{p \in [0\; 1]}U(p)}{\min_{p^* \in S }\lambda U(p^*,p^*)},
\end{equation}
where $S \subset [0\;1]$ denotes the set of symmetric Nash equilibria.

We study the PoA as a function of $\rho$ for a given value of $\lambda$ and $C$. The utilities studied in Section \ref{sec:UtilityWithGoodput} and \ref{sec:UtilityWithDelay} are considered below.

\subsection{Goodput}
In this subsection we consider the utility defined in Section \ref{sec:UtilityWithGoodput}. Let us begin by considering the team utility.  When all the nodes use the MAP $p$, then from Equation (\ref{eqn:UtilitySuccessRate}), the team utility is given by
\begin{equation}
\label{eqn:PoAGlobalGoodputUtility}
U(p)=\lambda p \exp\{-p\lambda C\}-\lambda p \rho.
\end{equation}
If $\rho >1$, the maximum value of the utility is zero and the maximum is attained at $p_m=0$. Let $p_m:=p_m(\rho)$ denote the MAP that maximizes the  team utility in (\ref{eqn:PoAGlobalGoodputUtility}). The following lemma gives its value.
\begin{lemma}
\label{lma:PoATeamGoodputOptimizer}
The MAP value that maximizes the team utility (\ref{eqn:PoAGlobalGoodputUtility}) is given by \[p_m=\frac{1-\mathbb{W}(\rho e)}{\lambda C}\]
for all $\rho \geq 0$ if $\lambda C \geq 1$, and if $\lambda C < 1$ it is the maximizer for $\rho$ such that $\mathbb{W}(\rho e)\geq 1-\lambda C.$
Further the maximum team utility is given by
$U(p_m)=\frac{\rho (1-\mathbb{W}(\rho e))^2}{ C \mathbb{W}(\rho e)}.$
\end{lemma}


By using the definition of the Lambert function, one can verify that $U(p_m)$ is a decreasing function in $\rho$. Indeed, differentiating $U(p_m)$ with respect to $\rho$ we have

\begin{eqnarray}
\label{eqn:PoAGoodPutDeriGlobalUtility}
\frac{\partial U(p_m)}{\partial \rho}&=&\frac{(1-\mathbb{W}(\rho e))^2}{C \mathbb{W}(\rho e)}
-\rho e \frac{(1-\mathbb{W}(\rho e))}{C \mathbb{W}(\rho e)}\frac{(1+\mathbb{W}(\rho e))\mathbb{W}^\prime(\rho e)}{C \mathbb{W}(\rho e)}\\
\label{eqn:PoAGoodPutSubsLambertDerv}
&=&\frac{(1-\mathbb{W}(\rho e))^2}{C \mathbb{W}(\rho e)} -\frac{(1-\mathbb{W}(\rho e))}{C \mathbb{W}(\rho e)}.
\end{eqnarray}
In (\ref{eqn:PoAGoodPutDeriGlobalUtility}) $\mathbb{W}^\prime$ denotes the derivative of the Lambert function. Equation (\ref{eqn:PoAGoodPutSubsLambertDerv}) follows by applying formula for the derivative of Lambert function. The last equation is negative valued for all $\rho \in [0,\;1]$. Thus the optimal utility is a decreasing function in $\rho$.

Let us look at the utility of the tagged node at equilibrium. From Proposition \ref{prop:GoodputequilibriumMAP}
we have
\begin{equation}
\label{eqn:PoAGoodPutEquilibriumutility}
U(p^*,p^*)=\left\{
  \begin{array}{ll}
    \exp\{-\lambda C\}-\rho, & \hbox{if \;} \rho \leq \exp\{-\lambda C\}\\
    0, & \hbox{if \;} \rho \geq \exp\{-\lambda C\}
  \end{array}
\right.
\end{equation}
The utility at equilibrium is also a decreasing function in $\rho$ for all $\rho\leq \exp\{-\lambda C\}$. With the expression for utility at equilibrium and global optimum we have the following result for PoA:

\begin{thm}
The value of the PoA is
\begin{equation}
PoA(\rho)=\left\{
            \begin{array}{ll}
            \frac{\rho (1-\mathbb{W}(\rho e))^2}{ \lambda C \mathbb{W}(\rho e)\{\exp\{-\lambda C\}-\rho\}}, & \hbox{if \;}\rho < \exp\{-\lambda C\}\\
         \infty, & \hbox{if\;} \rho \geq \exp\{-\lambda C\}
            \end{array}
          \right.
\end{equation}
\end{thm}
The PoA is shown as a function of $\rho$ in Figure \ref{fig:PoAGoodput}. From this figure we see that as $\rho$ increases, the PoA grows unboundedly. Thus the PoA is optimal when the pricing factor is set to zero. If $\lambda C \geq 1$ then the PoA is infinite by definition at $\rho=1/e$. However we noted in Section \ref{sec:UtilityWithGoodput} that the optimal performance of the spatial density of success is achieved at equilibrium with the same price factor. If $\lambda C < 1$ then the PoA is infinite at $\rho=\exp\{-\lambda C\}$. But again we noted in Section \ref{sec:UtilityWithGoodput} that at this price factor the optimal performance of the spatial density of success is achieved at equilibrium.

In Figure \ref{fig:PoAGoodputEqOptMAP} the equilibrium MAP and global optimal MAP are shown. For all values of $\rho$, the equilibrium MAP is larger than the global optimal MAP. Hence the nodes transmit more aggressively at equilibrium. But we note from Figure \ref{fig:PoAGoodputEqOptMAP} that the gap between the global optimal MAP and the equilibrium MAP reduces with pricing.

\subsection{Delay}

Consider the utility function in Equation (\ref{eqn:DelayUtility}).
The team utility for this game, when each node transmits with MAP $p$ is
\begin{equation}
\label{eqn:DelaySocialUtility}
U(p)=\frac{-\lambda }{p\exp\{-p\lambda C\}} - \lambda \rho p.
\end{equation}
It is easy to verify that the above utility function is concave in $p$. Assume that $\lambda C >1$. Then the unique MAP, denoted as $p_m:=p_m(\rho)$ that maximizes the social utility  satisfies
\begin{equation}
\label{eqn:DelaySocialOptimalSoln}
\exp\{p_m\lambda C \}(1-p_m\lambda C)=\rho p^2.
\end{equation}
We obtain this by differentiating Equation (\ref{eqn:DelaySocialUtility}) and setting to zero. Note that any $p_m$ that satisfies Equation (\ref{eqn:DelaySocialOptimalSoln}) also satisfies $p_m\lambda C \leq 1$, hence $p_m \in [0\;1]$. Also, it can be easily verified that $p_m$ is decreasing in $\rho$.

The utility at equilibrium can be obtained by using the equilibrium MAP in (\ref{eqn:LambertSolution}) and (\ref{eqn:DealyAtEquilibrium}), as
\begin{equation}
\label{eqn:DelayUtilityAtEquilibrium}
U(p^*,p^*)=(-2\lambda \rho/\lambda \overline{C})\mathbb{W}\left(\frac{-\lambda \overline{C}}{\sqrt \rho}\right).
\end{equation}
When $\lambda C >1$, from Lemma \ref{lma:DealyTwoEquilibriumCondn}, two symmetric Nash equilibria are possible for the price factor $\rho \leq \rho_{-1}$. Hence the above utility function can take two values, one corresponding to each equilibrium. Recall that we denoted by $p^*_0$ the SNE computed on the principal branch, and by $p^*_{-1}$ that computed on the other branch of the Lambert function. Recall that
$ p^*_0 \leq p^*_{-1}.$ The following proposition gives a bound for the PoA

\begin{thm}
\label{prop:PoABounds}
For Poisson bipolar MANETs with utility as in(\ref{eqn:DelayUtility}),
\begin{eqnarray}
\label{eqn:DelayPoABounds}
\frac{p_m(\rho_{-1})\lambda (2-p_m(\rho_{-1})\lambda C)}{2(1-p_m(\rho_{-1})\lambda C)}\leq PoA(\rho)
\leq \frac{p_m(\rho_t)\lambda \overline{C}(2-p_m(\rho_t)\lambda C)}{2(1-p_m(\rho_t)\lambda C)}
\end{eqnarray}
for $\rho \in [\rho_t \;\rho_{-1}]$, where $\sqrt {\rho_0}=e\lambda \overline C$. In addition
\begin{equation}
\label{eqn:DelayPoAUpperWBranch}
\frac{p_m}{p^*_0}\leq PoA(\rho) \leq 1 \text{\;for \;} \rho \geq \rho_{-1}.
\end{equation}
\end{thm}

The PoA as a function of $\rho$ and the bounds obtained in Proposition \ref{prop:PoABounds} are shown in Figure \ref{fig:PoADelay}. The jump in the figure at $\rho=\rho_{-1}$ is due to two possible SNE for $\rho\leq \rho_{-1}$ and a unique SNE for $\rho >\rho_{-1}$. In the interval $[\rho_t, \;\;\rho_{-1}]$ the PoA is decreasing in $\rho$. This results from the bad Nash equilibrium that occurs on the $\mathbb{W}_{-1}$ branch of the Lambert function which increases in $\rho$. If the central agent can't set a price factor higher than $\rho_{-1}$ then, from the PoA point of view, it is desirable to set the lowest possible price factor, i.e., $\rho=\rho_t$. For a price factor larger than $\rho_{-1}$, there is unique SNE, which is smaller than the equilibrium that occurs on the $\mathbb{W}_{-1}$ branch and  decreases\footnote{Principle branch of Lambert function is decreasing function of $\rho$} with $\rho.$ 
Thus setting high a price factor leads to improved PoA.

If $\rho_{-1}< 4 e (\lambda \overline C)^2 $, from Proposition \ref{prop:DelayGlobalAndEquilibriumOptimal}, by setting $\rho=4 e (\lambda \overline C)^2$, we can obtain better performance at equilibrium and also good PoA. If $\rho_{-1}> 4 e (\lambda \overline C)^2 $ then by setting $\rho=4 e (\lambda \overline C)^2$ one obtains global optimal performance at equilibrium provided the nodes settle at an equilibrium that lies on the principal branch of the Lambert function, otherwise this price factor leads to a poor PoA.

\section{Conclusions}
\label{sec:Conclusion}
Geometric considerations play a very central role in wireless
communications, since the attenuation of wireless channels
strongly depend on the distance between transmitter and
receiver. Models that take into account the exact location
of mobiles are often too complex to analyze or to optimize.
Our objective in this paper is to model competition between
mobiles as a game in which the locations of players is given
by a Poisson point process.

More structured point processes can also be contemplated, for instance
exhibiting attraction (hot spots) or repulsion (more elaborate medium
access control than Aloha like e.g. CSMA). We leave the analysis of medium access
games under such point processes for future research.

The competition we considered in the paper was between individual mobiles
each taking its own selfish decisions. We saw that the equilibrium of the game
results in a more aggressive access (larger access probabilistically). We
studied further pricing, and identified pricing parameters that induce
an equilibrium achieving the social  optimal performance.
On the other hand we showed that the utility at equilibrium can be zero).

We plan in the future to study other games within this framework: for instance
games with finitely many operators, each taking decisions for all its subscribers.
In addition we shall study jamming games.

\bibliographystyle{IEEEtran}
\bibliography{JSAC}

\vspace{-0.1in}

\section*{appendices}
\vspace{-0.1in}
\subsection*{Proof of Lemma \ref{lma:DealyTwoEquilibriumCondn}}
\label{app:ProofDealyTwoEquilibriumCondn}
First consider the $\mathbb{W}_0$ branch of the Lambert function. As $\rho$ takes value in the interval $[(\lambda \overline{C}e)^2 \; \infty]$, $\mathbb{W}_0\left (\frac{-\lambda \overline{C}}{\sqrt \rho}\right)$ increases continuously from $-1$ to $0$. Thus (\ref{eqn:DelayUtilityFixedPoint}) has a solution in the interval $[0 \; 1]$ if $\lambda \overline{C} \geq 1$. This implies that equilibrium point exists on the $\mathbb{W}_0$ branch for all $\rho$, satisfying $\sqrt \rho \geq \lambda \overline{C}e$.

\noindent
The  $\mathbb{W}_{-1}\left (\frac{-\lambda \overline{C}}{\sqrt \rho}\right)$ branch decreases continuously from $-1$ to $-\infty$ as $\rho$ takes value in the interval $[(\lambda\overline{C}e)^2 \; \infty]$. This implies that there exists a $\sqrt \rho_{-1} \geq \lambda \overline{C}e$ such that $\mathbb{W}_{-1}\left(\frac{-\lambda \overline{C}}{\sqrt \rho_{-1}}\right))=-\lambda \overline{C}$, and for all $\rho$ such that $\sqrt \rho_{-1}\geq \sqrt \rho \geq \overline{C}e$ satisfies $-\mathbb{W}_{-1}\left (\frac{-\lambda \overline{C}}{\sqrt \rho}\right)\leq \lambda \overline{C}$, resulting in a $p \in [0\; 1]$ that is a solution of (\ref{eqn:DelayUtilityFixedPoint}). Hence there exists an equilibrium point on the $\mathbb{W}_{-1}$ branch for all $\rho$ satisfying $\sqrt\rho_{-1}\geq \sqrt\rho$. This concludes the proof.
\vspace{-0.1in}
\subsection*{Proof of Lemma \ref{lma:DelayUniqueNE}}
\label{app:ProofDelayUniqueNE}
As in Lemma \ref{lma:DealyTwoEquilibriumCondn} we can argue that on the $\mathbb{W}_0$ branch, there exists $\rho_0 \geq \lambda\overline C e$ such that $\mathbb{W}_0 \left (-\frac{\lambda \overline{C}}{\sqrt \rho_0}\right )=-\lambda \overline{C}$ and for all $\rho$ such that $\rho\geq \rho_{-1}\geq \lambda\overline{C}e$ satisfies $-\mathbb{W}_0\left (\frac{-\lambda \overline{C}}{\sqrt \rho}\right)\leq \lambda \overline{C}$ as $-\mathbb{W}_0$ is decreasing in $\rho$.
\vspace{-0.1in}
\subsection*{Proof of Lemma \ref{lma:DelayObjectiveProperties}}
\label{app:ProofDelayObjectiveProperties}
Differentiating $h(\rho)$ with respect to $\rho$
\begin{eqnarray}
\label{eqn:DelayEquilibriumDerivative1}
\lefteqn{\frac{\rm d}{\rm d \rho}h(\rho)
=-\mathbb{W}\left(-\lambda\overline C / \sqrt \rho \right)-  \mathbb{W}^\prime \left(-\lambda\overline C / \sqrt \rho \right)(\lambda\overline C /2  \sqrt \rho)}\\
\label{eqn:DelayEquilibriumDerivative2}
&=&-\mathbb{W}\left(-\lambda\overline C / \sqrt \rho \right)\left (1-\frac{1}{2(1+\mathbb{W}\left(-\lambda\overline C / \sqrt \rho \right))}\right).
\end{eqnarray}
In Equation (\ref{eqn:DelayEquilibriumDerivative1}) $\mathbb{W}^\prime(\cdot)$ denotes the derivative of the Lambert function which is given as \cite{LambertFunction}[eqn. 3.2]
\begin{equation}
\label{eqn:LambertDerivative}
\mathbb{W}^\prime(x)=\frac{\mathbb{W}(x)}{x(1+\mathbb{W}(x))} \;\;\text{for} \;\; x\neq 0,x\neq -1/e.
\end{equation}
Equation (\ref{eqn:DelayEquilibriumDerivative2}) is obtained by substituting the derivative in (\ref{eqn:LambertDerivative}), evaluated at $x=-\lambda\overline C /\sqrt \rho$, in Equation (\ref{eqn:DelayEquilibriumDerivative1}).
Recall that on the principal branch of the Lambert function $\mathbb{W}(-\lambda \overline C /\sqrt \rho)$ is a negative valued increasing function in $\rho$. Then the term within parenthesis in (\ref{eqn:DelayEquilibriumDerivative2}) is a increasing function of $\rho$ passing through the origin at $\rho^*$ that satisfies $\mathbb{W}(-\lambda \overline C /\sqrt {\rho^*})=-1/2$.  Thus $h(\rho)$ is decreasing for $\rho \leq \rho ^*$ and increasing for $\rho \geq \rho ^*$. From \cite{ConvexOptimizationBook}[sec. 3.4.2] we conclude that $h(\rho)$ is a quasi convex function in $\rho$.

\noindent
Further by the definition of the Lambert function
\begin{eqnarray}
\label{eqn:OptimalRhoCondition}
-\lambda \overline C/\sqrt {\rho^*} &=&\mathbb{W}(-\lambda\overline C/{\sqrt {\rho^*}})\exp\{\mathbb{W}(-\lambda\overline C/\sqrt {\rho^*})\}=-\frac{1}{2}\exp\{-1/2\}.
\end{eqnarray}
Rearranging Equation (\ref{eqn:OptimalRhoCondition}), we get
$\rho^*=4 e (\lambda\overline C)^2$.
The other part of the Lemma follows by noting that $-\mathbb{W}(-\lambda \overline C /\sqrt \rho)$ is an increasing function in $\rho$ on the $\mathbb{W}_{-1}$ branch.
\vspace{-0.1in}
\subsection*{Proof of Proposition \ref{prop:PoABounds}}
\label{app:ProofPoABounds}
From Equation (\ref{eqn:DelayUtilityAtEquilibrium}) and (\ref{eqn:DelaySocialUtility}) we have
\begin{eqnarray}
\label{eqn:DelayPoARatio}
PoA(\rho) &=&\frac{\exp\{p_m\lambda C\}/p_m + \rho p_m}{(-2\rho/\lambda \overline{C})\mathbb{W}\left(-\frac{\lambda \overline{C}}{\sqrt \rho}\right)}\\
\label{eqn:DelayPoASubsSocialOptimalCondn}
&=&\frac{p_m/(1-p_m\lambda C) +  p_m}{(-2/\lambda \overline{C})\mathbb{W}\left(-\frac{\lambda \overline{C}}{\sqrt \rho}\right)}\\
\label{eqn:DeayPoARearrange}
&=&\frac{p_m(2-p_m\lambda C)}{\left(-2/\lambda \overline{C}\right)\mathbb{W}\left(-\frac{\lambda \overline{C}}{\sqrt \rho}\right)(1-p_m\lambda C)}\\
\label{eqn:DelayPoALowerBound}
&\geq&\frac{p_m}{\left(-1/\lambda \overline{C}\right)\mathbb{W}\left(-\frac{\lambda \overline{C}}{\sqrt \rho}\right)}.
\end{eqnarray}
We arrive at equality (\ref{eqn:DelayPoASubsSocialOptimalCondn}) by dividing both numerator and denominator in (\ref{eqn:DelayPoARatio}) by $\rho$, and applying the relation in (\ref{eqn:DelaySocialOptimalSoln}). Equality (\ref{eqn:DeayPoARearrange}) is obtained by simple rearrangement of terms in the previous step.

\noindent
To derive the bounds in (\ref{eqn:DelayPoABounds}), we consider the equilibrium computed on the $\mathbb{W}_{-1}$ branch of the Lambert function as it leads to the worst case equilibrium. This equilibrium is an increasing function in $\rho$ in the interval $ \rho_0\leq \rho\leq \rho_{-1}$ as discussed in the proof of Lemma \ref{lma:DealyTwoEquilibriumCondn}. Also, recall that the value of $p_m$ is decreasing in $\rho$. Thus the numerator in (\ref{eqn:DelayPoASubsSocialOptimalCondn}) is decreasing in $\rho$. Which implies that the ratio in (\ref{eqn:DeayPoARearrange}) is also decreasing in $\rho$. The upper bound in (\ref{eqn:DelayPoABounds}) now follows by noting that $-\mathbb{W}(\rho_0)=1$. To obtain the lower bound we use the relation $-\mathbb{W}(-\lambda \overline C/\sqrt {\rho_{-1}})=\lambda \overline C$ in (\ref{eqn:DelayPoALowerBound}).

For values of $\rho$ larger than $\rho_{-1}$ the SNE is unique, resulting from the principal branch of the Lambert function. The upper bound in (\ref{eqn:DelayPoAUpperWBranch}) follows directly by the definition of PoA, and the lower bound follows from the Inequality (\ref{eqn:DelayPoALowerBound}). Note that the lower bound is a function of $\rho.$


\section*{Proof of Lemma \ref{lma:NEExistance}}
\label{app:ProofNEExistance}
Consider a point to set map $\delta:[0,\;1]\rightarrow [0,\;1] $ defined by
\[\delta (p)=\left \{p^\prime | U(p^\prime,p)=\max_{q \in [0,\;1]}U(q,p)\right \}.\]
This defines the set of best responses of
the tagged node when all the other nodes use the MAP $p$. It follows from the continuity of $U(p^\prime,p)$ in $p$ and concavity in $p^\prime$ for a fixed $p$, that $\Gamma$ is an upper continuous mapping that maps each point of the set $[0,\;1]$ into a subset of $[0,\;1].$ By the Kakutani fixed point theorem, there exists a point $p^* \in [0,\;1]$ such that
\[U(p^*,p^*)=\max_{p^\prime \in [0,\;1]} U(p^\prime,p^*).\]
Then $p^*$ is the symmetric Nash equilibria by Definition \ref{defn:SymmetricNashEquilibrium}.
\vspace{-0.1in}
\subsection*{Proof of Lemma \ref{lma:PoATeamGoodputOptimizer}}
\label{app:ProofPoATeamGoodputOptimizer}
Assume $\rho\leq 1$, then $p_m >0 $ and satisfies the relation $\exp\{-p_m\lambda C\}(1-p_m\lambda C)= \rho.$
By rearranging, $p_m$ can be expressed $p_m=(1-\mathbb{W}(\rho e))(\lambda C).$
Recall that $\mathbb{W}(\cdot)$ is a monotonically increasing function taking values $\mathbb{W}(0)=0$ and $\mathbb{W}(e)=1$. If $\lambda C \geq 1$ the $p_m$ lies in the interval $[0 ,\;1]$ for all $\rho \in [0 ,\;1]$. If $\lambda C <1$, then $p_m$ lies in the interval $[0 ,\;1]$ for all $\rho$ such that $\mathbb{W}(\rho e) \geq 1-C \lambda$. Thus whenever $\lambda C <1$ we assume that $\rho$ satisfies $\mathbb{W}(\rho e) \geq 1-C \lambda$. The optimal value of the utility function is

\begin{eqnarray}
\label{eqn:PoAGoodPutSubstituteGlobalOpt}
p_m\exp\{p_m\lambda C\}&=&\frac{1-\mathbb{W}(\rho e)}{\lambda C}\exp\{-1+\mathbb{W}(\rho e)\} \\
\label{eqn:PoAGoodPutApplyLambertDefn1}
&=& \frac{1}{e \lambda C}\{\exp\{\mathbb{W}(\rho e)\} - \rho e\}\\
\label{eqn:PoAGoodPutApplyLambertDefn2}
&=&\frac{1}{e \lambda C}\left \{\frac{\rho e}{\mathbb{W}(\rho e)}- \rho e\right \}
=\frac{\rho}{ \lambda C}\left \{\frac{1}{\mathbb{W}(\rho e)} -1 \right \}
\end{eqnarray}
where (\ref{eqn:PoAGoodPutSubstituteGlobalOpt}) is obtained by substituting the value of MAP maximizes the team utility. Equation (\ref{eqn:PoAGoodPutApplyLambertDefn1}) and (\ref{eqn:PoAGoodPutApplyLambertDefn2}) follows by application of the definition of the Lambert function. The maximum utility for the team case can be now computed as a function of $\rho$

\begin{equation}
\label{eqn:PoAGoodputGlobalOptima}
U(p_m)=\frac{\rho (1-\mathbb{W}(\rho e))^2}{ C \mathbb{W}(\rho e)}.
\end{equation}
%
%
\begin{table}
\caption{Characterization of symmetric Nash equilibria}
\centering
\begin{tabular}{c c c}
   $\lambda \overline C$ & Price factor & SNE $(p^*)$\\
   \hline\hline
    & $\rho \leq \rho_0$ & $1$ \\
                 \raisebox{1.5ex}{$\lambda \overline C <1$}     & $\rho > \rho_0$ & $-\mathbb{W}_0(-\lambda \overline C / \sqrt \rho)/\lambda \overline C$ \\
   \hline
  & $\rho < (e\lambda \overline C)^2$ & $1$ \\
  \raisebox{1.2ex}{$\lambda \overline C  \geq 1$} &  $\rho_{-1} \geq \rho \geq  (e\lambda \overline C)^2$ &  $-\mathbb{W}_0(-\lambda \overline C / \sqrt \rho)/\lambda \overline C \text{\; or \;} -\mathbb{W}_{-1}(-\lambda \overline C / \sqrt \rho)/\lambda \overline C$ \\
              & $\rho \geq \rho_{-1}$  &  $-\mathbb{W}_0(-\lambda \overline C / \sqrt \rho)/\lambda \overline C$ \\
  \hline
\end{tabular}
\label{tab:DelaySNE}
\end{table}

\begin{figure}[!htb]
\centering
  \includegraphics[scale=.4]{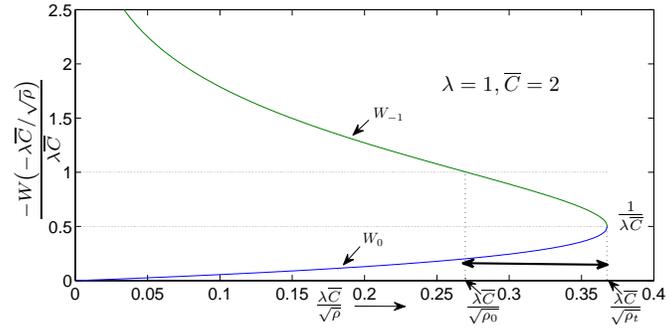}\\
  \caption{\small{SNE on both branches}}\label{fig:DelayLambert1}
\end{figure}

\begin{figure}[!htb]
\centering
  \includegraphics[scale=.4]{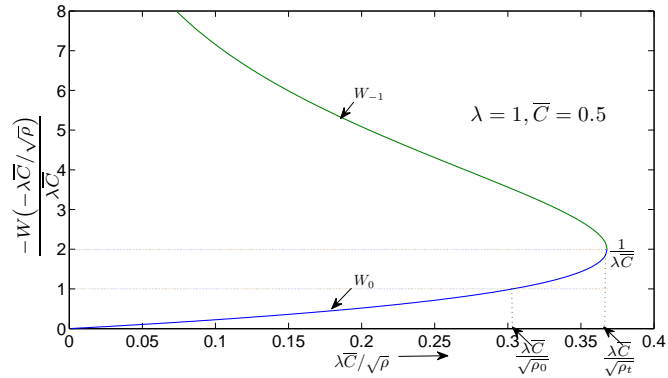}\\
  \caption{\small{Unique SNE }}\label{fig:DelayLambert2}
\end{figure}

\begin{figure}[!htb]
\centering
  \includegraphics[scale=.4]{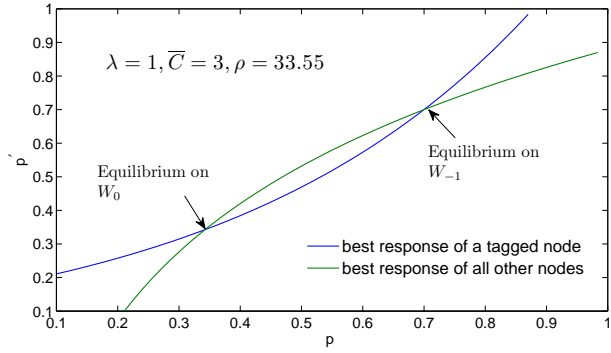}\\
  \caption{\small{Best Response}}\label{fig:Best Response}
\end{figure}


%
%

\begin{figure}[h]
\begin{centering}
  \includegraphics[scale=.4]{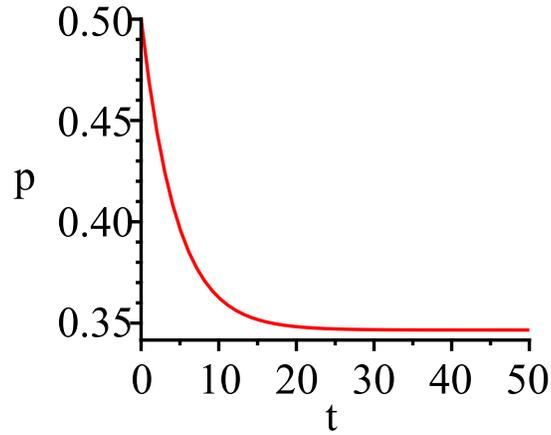}\\
  \caption{Convergence}\label{fig:Dynamics}
  \end{centering}
\end{figure}

\begin{figure*}[!htb]
 \begin{minipage}{0.4\linewidth}
    \centering
   \includegraphics[scale=.4]{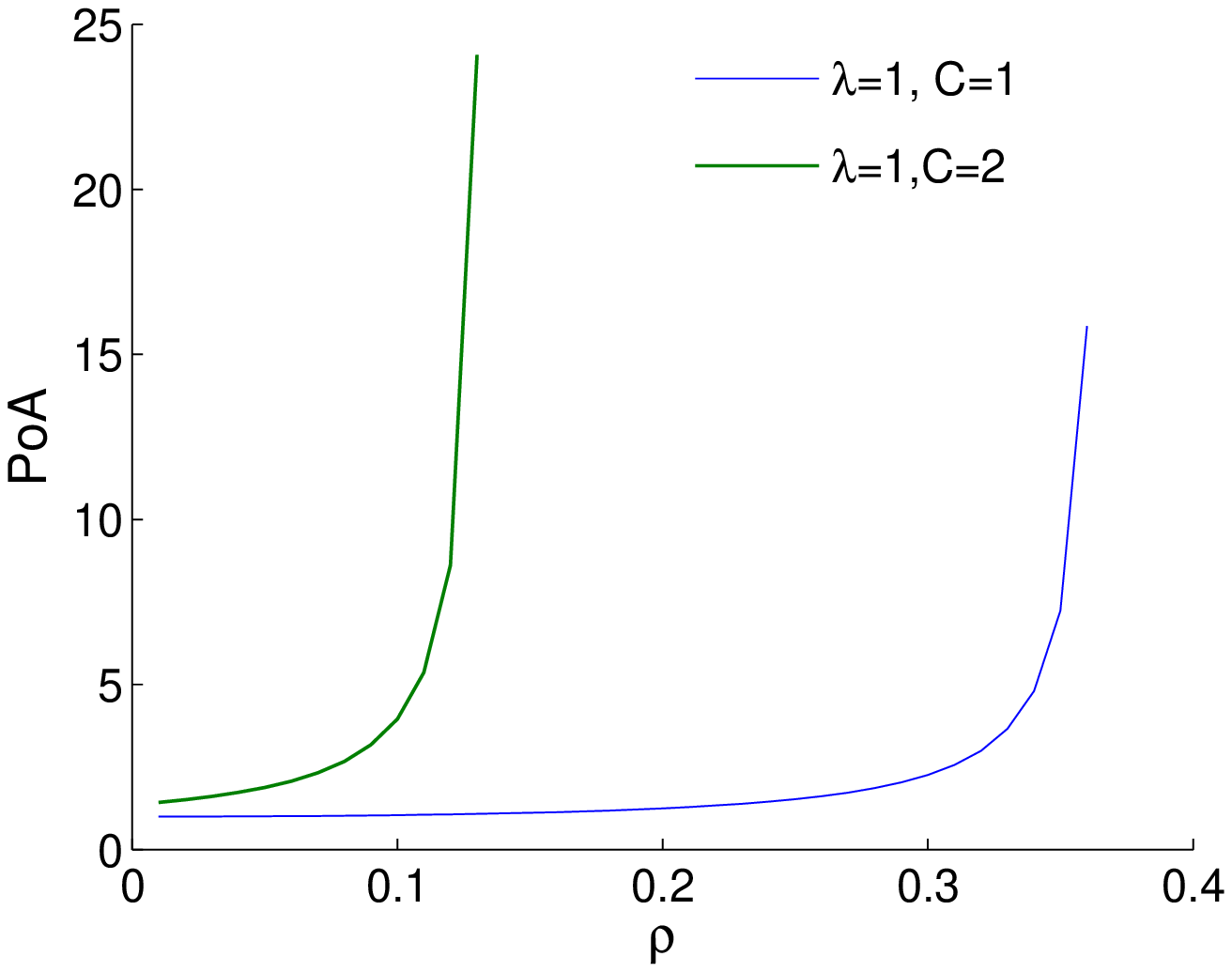}
   \caption{\small{PoA for Goodput}}
   \label{fig:PoAGoodput}
 \end{minipage}
 \hspace{0.5cm}
 \begin{minipage}{0.4\linewidth}
    \centering
   \includegraphics[scale=.4]{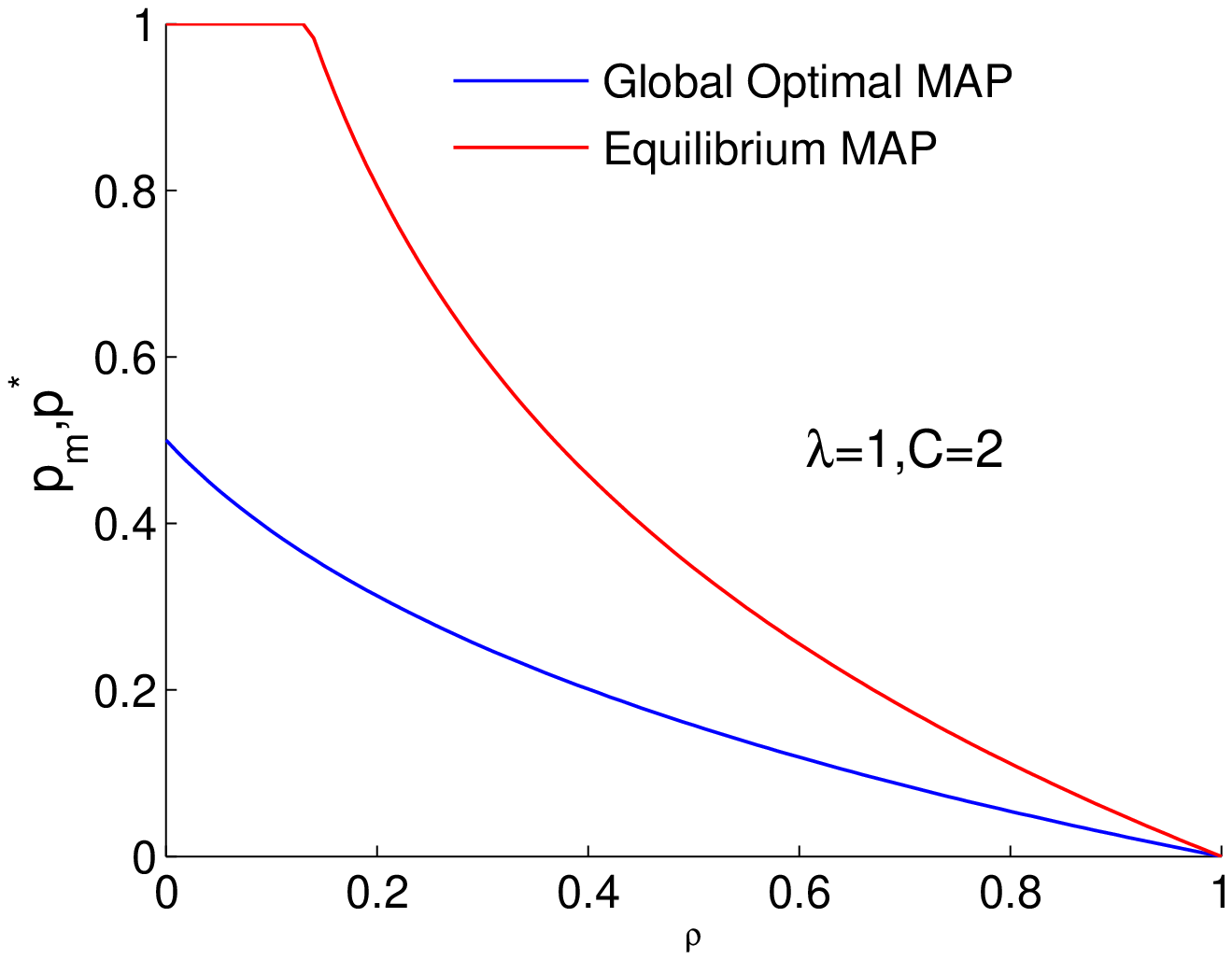}
   \caption{\small{Equilibrium and Optimal MAP}}
   \label{fig:PoAGoodputEqOptMAP}
 \end{minipage}
\end{figure*}


\begin{figure}
\centering
  \includegraphics[scale=.4]{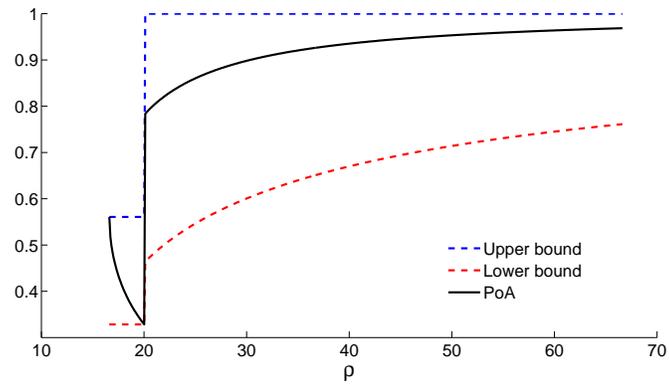}\\
  \caption{\small{PoA for delay, C=3, $\lambda$=1}}
   \label{fig:PoADelay}
\end{figure}

\end{document}